\documentclass[letterpaper]{article} 
\usepackage[preprint]{aaai2027}
\usepackage[hyphens]{url}  
\usepackage{graphicx} 
\usepackage{natbib}  
\usepackage{caption} 
\usepackage{algorithm}

\usepackage{newfloat}
\usepackage{listings}
\DeclareCaptionStyle{ruled}{labelfont=normalfont,labelsep=colon,strut=off} 
\floatstyle{ruled}
\newfloat{listing}{tb}{lst}{}
\floatname{listing}{Listing}

\usepackage{booktabs}

\usepackage{textcomp}

\usepackage{comment}
\usepackage{amssymb,amsfonts}

\usepackage{url}
\usepackage[T1]{fontenc}
\usepackage[utf8]{inputenc}
\usepackage{blindtext}

\usepackage{color}
\usepackage{here}
\usepackage{fancyhdr}
\usepackage[title,titletoc]{appendix}

\usepackage{nameref}
\usepackage{subcaption}
\usepackage{lineno}

\usepackage{bm,enumerate,amsmath,amssymb,amsthm}
\newtheoremstyle{assumption}
  {\topsep}
  {\topsep}
  {}
  {1em}
  {\itshape}
  {}
  {.5em}
  {\thmname{#1}\thmnumber{ #2}\thmnote{ (#3)}}
\theoremstyle{assumption}
\newtheorem{assumption}{Assumption}
\newtheorem{lemma}{Lemma}
\newtheorem{remark}{Remark}

\newtheorem{theorem}{Theorem}
\newtheorem{example}{Example}
\newtheorem{definition}{Definition}

\newtheorem{proposition}{Proposition}

\usepackage[noend]{algpseudocode} 
\algnewcommand\Input{\item[{\textbf{Input:}}]}
\algnewcommand\Output{\item[{\textbf{Output:}}]}
\algrenewcommand\algorithmicdo{}
\algnewcommand\algorithmicto{\textbf{to}}
\algnewcommand\algorithmicbreak{\textbf{break}}
\algdef{SE}[FOR]{ForTo}{EndFor}[2]{\algorithmicfor\ #1\ \algorithmicto\ #2\ \algorithmicdo}{\algorithmicend\ \algorithmicfor}%
\algtext*{EndFor}
\algdef{SE}[IF]{NoThenIf}{EndIf}[1]{\algorithmicif\ #1}{\algorithmicend\ \algorithmicif}%
\algtext*{EndIf}
\let\oldReturn\Return
\renewcommand{\Return}{\State\oldReturn}

\newcommand{\mc}{\mathcal}

\newcommand{\newcolor}{black}

\def\arxiv{0}

\if\arxiv0

\else

\fi

\title{Computing Stable Matchings under Complementarities and Preference Misalignment}

\author{
    Tatsuya Iwase\textsuperscript{\rm 1},
    Baharak Rastegari\textsuperscript{\rm 2},
    Sebastian Stein\textsuperscript{\rm 2},
    Enrico H. Gerding\textsuperscript{\rm 2}
}
\affiliations{
    \textsuperscript{\rm 1}Toyota CRDL Inc., Aichi, Japan\\
    tiwase@mosk.tytlabs.co.jp\\
    \textsuperscript{\rm 2}University of Southampton, Southampton, United Kingdom\\
    B.Rastegari@soton.ac.uk, \{ss2, eg\}@ecs.soton.ac.uk

}

\begin{document}

\maketitle


\begin{abstract}
We study many-to-one, two-sided stable matching problems in which preferences are complementary and firms and workers may rank the same allocation differently. Here a \emph{coalition} is a group of workers that can be jointly matched with a firm. With complementary preferences, a stable matching need not exist. We show that a stable matching exists for every instance when two conditions hold. The first requires each firm to have an \emph{anchor worker} who is included in every feasible coalition that can be matched with that firm. The second is Transitive Alignment, which means there must be an overall ranking across different coalitions that is consistent with the workers' preferences.
We propose CDAR (Combinatorial Deferred Acceptance with Reproposals),
a Deferred Acceptance-style algorithm that permits reproposals, and prove its
finite-time convergence and its output of a stable matching by constructing an
$N$-digit potential function.
Moreover, we show that, under strict preferences, the output of CDAR is Pareto
optimal among stable matchings.
The proposed framework naturally captures applications such as vehicle-route
assignment for traffic safety and the misaligned interests of labor and
management in wage negotiation.
\end{abstract}

\section{Introduction}
\label{sec:intro}

In many settings, such as multi-agent services, transportation, and supply chains, the problem arises of assigning a group of multiple heterogeneous agents to jointly accomplish a given task.
The generalization of such problems is the many-to-one two-sided matching
problem, and in this paper, following the convention of matching theory, we
call the two sides \emph{workers} and \emph{firms}, respectively.
In many-to-one matching theory, the vast majority of prior work has assumed \emph{substitutable} preferences. Substitutability is the property that workers can serve as substitutes for one another; under this assumption, the existence of stable matchings and the convergence of the Deferred Acceptance (DA) algorithm are known to be guaranteed \cite{gale1962college,kelso1982job,hatfield2005matching,ostrovsky2008stability,roth1990two}.

In contrast, many real-world assignment problems exhibit \emph{complementarity}, whereby workers need one another. A representative example is the \emph{peer effect}, in which a worker's utility depends on the composition of the other workers employed by the same firm. The same structure appears in several settings. In team formation, the value of a team improves only once a particular member joins. In transportation, which motivates this work, a trip is realized only if every segment of the route is secured. In service chaining for distributed AI services, the service is realized only if all processing steps are executed. Under such complementary preferences, the classical theory built on substitutability is generally not applicable, and,
as exemplified by resident matching with couples \cite{roth1984evolution,biro2013matching},
even the existence of a stable matching is not guaranteed. Many recent approaches to complementarity consider settings in which workers
rank only firms (excluding peer effects)
\citep{che2019stable,nguyen2018near,nguyen2021stability,rostek2020matching,huang2023stable},
and do not guarantee exact stability in finite markets under peer effects.

Among these lines of work, the one that squarely addresses peer effects is \citet{pycia2012stability}. That paper proposed a framework in which both firms and workers hold general preferences over entire employment coalitions.
The central result of \citet{pycia2012stability} is that a stable matching
exists whenever the preferences induced by a \emph{division rule} (e.g., Nash
bargaining) that allocates the value of each coalition satisfy \emph{pairwise
alignment} (PA)---the property that any two agents who belong to two common
coalitions always rank those coalitions in the same order.
However, such a framework, which aligns the rankings of the two sides through
the division of value, cannot handle situations in which the two sides'
evaluations are given independently.
For instance, in wage negotiation, the workers desire high wages while the firm
wishes to keep its payments low, so their evaluations of the same coalition are
reversed and pairwise alignment fails. As a
consequence, applications in which the two sides' evaluations are not aligned
have remained outside the scope of existing theory.

The framework of \citet{pycia2012stability} has a further limitation. While it guarantees \emph{that a stable matching exists}, it provides no \emph{procedure for actually constructing} such a matching, and it is not clear how an algorithm for finding a stable matching under pairwise-aligned preferences should be designed. This gap between existence proofs and constructive algorithms is a major
obstacle to implementation in real problems.
By contrast, for one-sided hedonic games, constructive procedures are available.
Here the core---a partition admitting no blocking coalition---is the counterpart
of stability, and greedy methods for computing the core \citep{caskurlu2024hedonic}
and the convergence of deviation dynamics \citep{brandt2023reaching} have been
studied. However, these methods do not extend to two-sided settings. Moreover,
the greedy method presumes a global ranking shared by all agents together with
explicit handling of coalitions, which is infeasible when the number of
coalitions is exponential. In the two-sided setting,
\citet{echenique2007solution} enumerates the strong core \citep{roth1990two} but
does not guarantee its existence.

In this paper, to resolve these issues, we generalize the framework of
\citet{pycia2012stability} to the firm-worker misalignment setting and provide
a proposal-based algorithm that always finds a stable matching even when the
strong core is empty.
To this end, we introduce a mild assumption that any instance can be made to
satisfy: for each firm there exists a worker---the \emph{anchor worker}---
contained in every feasible coalition of that firm. This assumption serves as one of the constraints that align firms and workers whose preference orders need not coincide. Since the anchor condition alone does not function as a sufficient condition for stability, we propose a complementary condition on worker preferences called \emph{Transitive Alignment} (TA).
TA requires the existence of a metaranking, common to the workers, over the
desirability of coalitions.
While TA implies the PA of \citet{pycia2012stability} restricted to comparisons
among workers, it imposes no restriction on firm utilities and is therefore a
weaker condition than PA. That is, the framework of this paper covers the class
targeted by \citet{pycia2012stability}.
Combining the anchor condition with TA yields a sufficient condition guaranteeing the existence of a stable matching even in the firm-worker misalignment setting. In addition, using an $N$-digit potential function, we prove that a Deferred Acceptance-style algorithm permitting reproposals converges in finite time and outputs a stable matching. This work is a first step toward providing a means of constructing solutions for stable matching problems under complementary preferences.

\paragraph{Contributions.} The contributions of this paper are as follows.

\begin{itemize}
    \item We generalize the framework of \citet{pycia2012stability} to the firm-worker misalignment setting, in which the preference orders of firms and workers need not coincide.
    \item We provide sufficient conditions---the anchor condition and
    Transitive Alignment (TA)---for the existence of a stable matching under
    firm-worker misalignment.
    \item We propose CDAR, an algorithm that always outputs a stable matching
    under the anchor condition and TA, using only local operations---without
    enumerating coalitions explicitly and without constructing the metaranking.
    \item By constructing an $N$-digit potential function, we prove that CDAR converges in finite time and outputs a stable matching.
    \item We show, by a reduction, that CDAR outputs a pairwise stable matching also in the setting of \citet{pycia2012stability}.
    \item Under strict preferences, we extend the known result that any stable matching is Pareto optimal within the set of stable matchings \citep{roth1990two} to the setting that permits complementarity.
    \item To demonstrate the applicability of the proposed method, we confirm
    that the anchor condition and TA hold by construction in the traffic
    application and that CDAR outputs stable matchings.
\end{itemize}

\section{Model}
\label{sec:model}

In this section, we formalize a model that generalizes the framework of \citet{pycia2012stability} to the firm-worker misalignment setting. To convey the intuitive correspondence, we introduce each concept below through the concrete example of ``a firm hiring a team of workers.''

\subsection{Agents, Coalitions, and Utilities}
\label{sec:notation}

Let $\mc{F}$ denote the set of firms ($N=|\mc{F}|$), $\mc{W}$ the set of workers,
and $\mc{I}=\mc{F}\cup\mc{W}$ the set of all agents. A subset of workers $C\subseteq\mc{W}$ is called a \emph{coalition} and represents a team hired by a firm. Each firm $f\in\mc{F}$ has a finite set of feasible coalitions $\mc{C}_f\subseteq 2^{\mc{W}}$, and the situation in which firm $f$ employs a coalition $C\in\mc{C}_f$ is called an \emph{allocation} $(f,C)$.

A firm's preferences are represented by a utility function over entire coalitions
\[
U_f:\mc{C}_f\cup\{\emptyset\}\to\mathbb{R}_{\ge0},
\]
with the utility of remaining without employees defined as $U_f(\emptyset):=0$. On the other hand, the preferences of a worker $w$ are represented by a utility function over the allocations to which the worker belongs,
\[
V_w:
\{(f,C)\mid f\in\mc{F},\,C\in\mc{C}_f,\,w\in C\}
\cup\{\emptyset\}
\to\mathbb{R}_{\ge0},
\]
with the utility of unemployment defined as $V_w(\emptyset):=0$. That is, a worker's utility varies depending on who the members (peers) of the coalition to which the worker belongs are, thereby expressing the peer effect. The essential difference from \citet{pycia2012stability} is that the firm utility $U_f$ and the worker utility $V_w$ are given independently. For instance, the case of multi-LLM agents---where each agent $w$ consumes
$V_w(f,C)$ tokens and the firm $f$ pays a budget $B(C):=\sum_{w\in C}V_w(f,C)$---is
a zero-sum-like situation with $U_f(C)\approx\text{const}-B(C)$.

Both utilities are assumed to satisfy individual rationality in the following sense:
\[
U_f(C)>U_f(\emptyset)=0,
\qquad
\forall f\in\mc{F},\,C\in\mc{C}_f,
\]
\[
V_w(f,C)>V_w(\emptyset)=0,
\qquad
\forall w\in C.
\]


For each firm $f\in\mc{F}$, the set of feasible coalitions $\mc{C}_f$ is finite,
and all utility values $U_f(C)$, $V_w(f,C)$ are assumed to be finite real numbers.

\subsection{Matching}
\label{sec:matching}

\begin{definition}[Matching]
\label{def:matching}
A matching is a mapping
\[
\mu:\mc{I}\to 2^{\mc{I}}
\]
satisfying the following:
\begin{itemize}
    \item $\mu(w)=\{f\}$ if worker $w$ is employed by firm $f$, and $\mu(w)=\emptyset$ if $w$ is unemployed,
    \item $\mu(f)\in\mc{C}_f\cup\{\emptyset\}$, $\forall f\in\mc{F}$,
    \item $\mu(w)=\{f\}\iff w\in\mu(f)$, $\forall w\in\mc{W},\,f\in\mc{F}$.
\end{itemize}
\end{definition}

That is, a matching $\mu$ assigns to each firm at most one coalition, and specifies each worker's employer consistently with the firm employing the coalition to which the worker belongs.


\subsection{Stability and Blocking Pairs}

Given a matching $\mu$, if a firm $f$ and a set of workers can jointly deviate from their current allocations and improve their payoffs, the matching is said to be unstable.

\begin{definition}[Blocking pair]
\label{def:blocking}
For $f\in\mc{F}$ and $C\in\mc{C}_f$, the allocation $(f,C)$ is a \emph{blocking pair} for a matching $\mu$ iff
\begin{align}
U_f(C) &> U_f(\mu(f)), \label{eq:block_firm}\\
V_w(f,C) &> V_w(\mu(w),\mu(\mu(w))),
\quad \forall w\in C \label{eq:block_worker}
\end{align}
hold. Here, when $w$ is unemployed ($\mu(w)=\emptyset$), the composition $\mu(\mu(w))=\mu(\emptyset)$ is not formally defined since the domain of $\mu$ is $\mc{I}$. In this paper we adopt the convention
\[
V_w(\mu(w),\mu(\mu(w))) := V_w(\emptyset)
\]
in that case. 
\end{definition}

Intuitively, when firm $f$ prefers $C$ to its current coalition $\mu(f)$ (Eq.~\eqref{eq:block_firm}), and every worker in $C$ prefers the allocation $(f,C)$ under firm $f$ to their current employment (Eq.~\eqref{eq:block_worker}), the pair $(f,C)$ constitutes a mutually desirable deviation.

\begin{definition}[Stability]
\label{def:stable}
A matching $\mu$ is \emph{stable} iff it admits no blocking pair. 
\end{definition}

We denote by $\mc{M}$ the set of all stable matchings. The problem of finding a
matching $\mu\in\mc{M}$ is called \emph{Misaligned Coalitional Matching}
(\textbf{MCM}).

\subsection{Anchor Assumption}
\label{sec:anchor}

Under complementary preferences, stable matchings may not exist in general \cite{roth1985college}. In this paper, we introduce the following structural assumption for each firm.

\begin{assumption}[Anchor worker]
\label{asmp:anchor}
For each firm $f\in\mc{F}$, there exists a worker $a_f\in\mc{W}$ such that
\[
a_f\in C,
\qquad
\forall C\in\mc{C}_f.
\]
That is, $a_f$ is contained in every feasible coalition of firm $f$. We call $a_f$ the \emph{anchor worker} of firm $f$.
\end{assumption}

The anchor worker makes any two feasible coalitions of firm $f$ comparable under $\succeq_\chi$, and plays a central role in the convergence proof given later (the monotonicity of the potential function).

To preserve generality, we do not exclude the case in which different firms
share the same anchor. In that case some firms may fail to obtain any coalition
($\mu(f)=\emptyset$), but the stability and convergence results below still hold.

\subsection{Transitive Alignment and Stability Conditions}
\label{sec:alignment}

The anchor condition on firms is not sufficient by itself to guarantee the existence of a stable matching. An additional consistency condition on workers’ preferences is also needed.

\citet{pycia2012stability} showed that for stable matchings to exist broadly
under complementary preferences, a common ranking that allows cross-agent
comparison of preferences (a \emph{metaranking}) is necessary.
Intuitively, under complementarity, the formation of one coalition can
determine the success or failure of another, so reaching coalition formation
requires a common yardstick that prioritizes among coalitions.
In the firm-worker misalignment setting of this paper, the preference orders
of firms and workers need not coincide, so a ranking common to all agents does
not hold in general.
We therefore introduce Transitive Alignment as a condition that requires the
existence of a metaranking \emph{restricted to the worker side}.
Below we write $\mc{A}:=\{(f,C)\mid f\in\mc{F},\ C\in\mc{C}_f\}$ for the set of
all allocations, and for an allocation $x=(f,C)$ we abbreviate $w\in C$ as
$w\in x$.

\begin{definition}[Transitive Alignment (TA)]
\label{def:ta}
The workers' utilities are \emph{transitively aligned} if there exists a weak
order (\emph{metaranking}) $\succeq_\chi$ on $\mc{A}$ such that for every
worker $w$ and all allocations $x,y\in\mc{A}$ with $w\in x$ and $w\in y$,
\[
V_w(x)\ \ge\ V_w(y)
\iff
x \succeq_\chi y
\]
holds.
\end{definition}

That is, TA requires the existence of a common cross-coalition ranking with
which each worker's preferences agree---over the allocations to which the
worker belongs.
In particular, on any two allocations sharing common workers, the orders of
all shared workers coincide via $\succeq_\chi$, so TA implies the condition
obtained by restricting the pairwise alignment (PA) of
\citet{pycia2012stability} to comparisons among workers
(Example~\ref{ex:no_ta}).
On the other hand, the order of any pair of allocations with no common worker
is unconstrained.
Moreover, since the firm utilities $U_f$ are given independently of
$\succeq_\chi$, TA is weaker than PA in that it does not require the
firm-worker agreement of preferences.
The existence of a metaranking is equivalent to the absence of strictly
improving cycles when the workers' preference comparisons are chained
\citep{suzumura1976remarks}, and for finite instances it can be decided and
constructed in polynomial time.
Furthermore, when the utilities of all workers can be written as $V_w(x)=g(x)$
for some function $g:\mc{A}\to\mathbb{R}$, $g$ itself is a metaranking and TA
holds trivially.
TA, however, requires only the existence of a metaranking; it does not require
that $g$ or $\succeq_\chi$ be given explicitly.

The following two examples show that neither the anchor condition nor TA alone can guaranttee the existence of a stable matching. The second example also shows that worker-side PA cannot subsittute for TA.

\begin{example}[TA holds, no anchor]
\label{ex:no_anchor}
Let the set of firms be $\mc{F}=\{A,B\}$ and the set of workers $\mc{W}=\{a,b,c\}$.
The feasible coalitions of firm $A$ are $\mc{C}_A=\{\{a,c\}\}$,
and those of firm $B$ are $\mc{C}_B=\{\{a\},\{b,c\}\}$.
The firm utilities are $U_A(\{a,c\})=1$, $U_B(\{a\})=2$, $U_B(\{b,c\})=1$.
The worker utilities are such that:
\begin{align*}
&V_a(A,\{a,c\}) > V_a(B,\{a\}) > V_a(\emptyset),\\
&V_b(B,\{b,c\}) > V_b(\emptyset),\\
&V_c(B,\{b,c\}) > V_c(A,\{a,c\}) > V_c(\emptyset).
\end{align*}
Since no worker is contained in all coalitions of firm $B$, the anchor condition fails.
Note that, if firm $B$ employs $\{a\}$, then $(A,\{a,c\})$ is a blocking pair;
if firm $A$ employs $\{a,c\}$, then $(B,\{b,c\})$ is a blocking pair;
and if firm $B$ employs $\{b,c\}$, then $(B,\{a\})$ is a blocking pair.
Since this cycle contains all feasible coalitions, no stable matching exists
(Fig.~\ref{fig:TAnoAnchor}).
TA, on the other hand, holds:
$(B,\{b,c\})\ \succ_\chi\ (A,\{a,c\})\ \succ_\chi\ (B,\{a\})$ gives a metaranking.
\end{example}


\begin{figure}[t]
\centering
\begin{subfigure}[t]{.54\columnwidth}
  \centering
  \includegraphics[width=.8\linewidth]{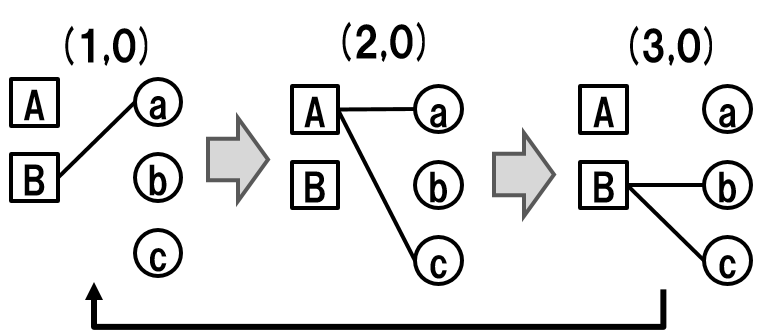}
  \subcaption{TA holds, no anchor (Example~\ref{ex:no_anchor}).
  The value at each state is the potential $\phi$ defined later.}
  \label{fig:TAnoAnchor}
\end{subfigure}\hfill
\begin{subfigure}[t]{.42\columnwidth}
  \centering
  \includegraphics[width=1.1\linewidth]{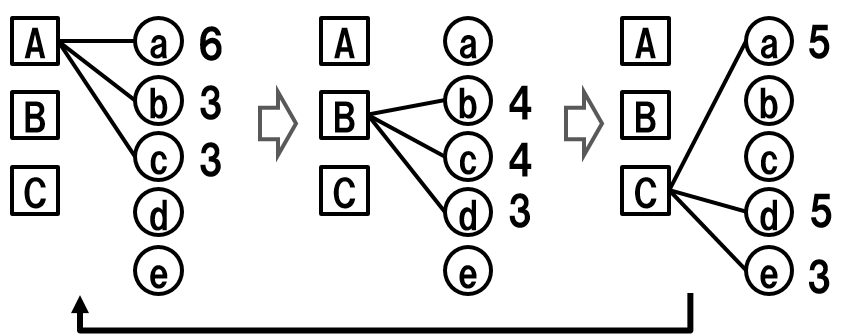}
  \subcaption{Anchor holds, PA holds, TA fails (Example~\ref{ex:no_ta})}
  \label{fig:PAnotTA}
\end{subfigure}
\caption{Two counterexamples in which no stable matching exists.
In (a), TA holds, so a metaranking and $\phi$ are well-defined, and $\phi$
decreases at the step where the anchor condition fails.
In (b), the workers' preferences are cyclic, so no metaranking exists at all.}
\label{fig:counterexamples}
\end{figure}

\begin{example}[Anchor holds, worker-side PA holds, TA fails]
\label{ex:no_ta}
Let the set of firms be $\mc{F}=\{A,B,C\}$ and the set of workers $\mc{W}=\{a,b,c,d,e\}$.
The feasible coalitions are $C_1=\{a,b,c\}$, $C_2=\{b,c,d\}$, $C_3=\{a,d,e\}$, with
$\mc{C}_A=\{C_1\}$, $\mc{C}_B=\{C_2\}$, $\mc{C}_C=\{C_3\}$.
All firm utilities are $U_A(C_1)=U_B(C_2)=U_C(C_3)=1$.
The anchor workers of the firms are $a_A=a$, $a_B=b$, $a_C=e$, respectively, so the anchor condition holds.
The worker utilities are as follows:
\begin{align*}
&V_a(A,C_1)=6,\quad V_a(C,C_3)=5,\\
&V_b(B,C_2)=4,\quad V_b(A,C_1)=3,\\
&V_c(B,C_2)=4,\quad V_c(A,C_1)=3,\\
&V_d(C,C_3)=5,\quad V_d(B,C_2)=3,\\
&V_e(C,C_3)=3.
\end{align*}
For $b,c \in C_1\cap C_2$ we have $V_b(B,C_2)>V_b(A,C_1)$ and $V_c(B,C_2)>V_c(A,C_1)$,
so PA among workers holds.
However, any metaranking $\succeq_\chi$ would have to satisfy simultaneously
$(B,C_2)\succ_\chi(A,C_1)$ by the preferences of $b$,
$(A,C_1)\succ_\chi(C,C_3)$ by the preferences of $a$, and
$(C,C_3)\succ_\chi(B,C_2)$ by the preferences of $d$,
which is cyclic and hence contradictory.
Therefore no metaranking satisfying TA exists.
Since each firm's set of feasible coalitions consists of only one coalition,
$C_1$, $C_2$, or $C_3$ respectively, the cycle $C_1\to C_2\to C_3\to C_1$ contains all feasible coalitions.
Therefore, whichever matching is chosen, the next coalition in the cycle forms a blocking pair,
and no stable matching exists (Fig.~\ref{fig:PAnotTA}).
\end{example}


\begin{remark}
\label{rem:wlog}
The anchor condition (Assumption~\ref{asmp:anchor}) can always be satisfied for
any instance by adding to each coalition in $\mc{C}_f$ a dummy worker
$d_f\notin\mc{W}$ dedicated to firm $f$. Moreover, if $d_f$'s preferences are
assumed to agree with the metaranking $\succeq_\chi$, then TA is preserved and
no stable matching of the original instance is lost.
\end{remark}

Imposing both conditions simultaneously guarantees the existence of a stable matching.

\begin{theorem}
\label{thm:existence}
Under the anchor condition (Assumption~\ref{asmp:anchor}) and TA
(Definition~\ref{def:ta}), every instance of MCM admits
at least one stable matching.
\end{theorem}

The proof of Theorem~\ref{thm:existence} is given by showing that the algorithm
CDAR proposed in the next section converges in finite time and outputs a stable
matching (Theorem~\ref{thm:converge}, Theorem~\ref{thm:stable}).
Note that the anchor condition and TA are sufficient conditions
for the existence of a stable matching but not necessary, as the following shows.

\begin{example}[No TA, no anchor, yet stable]
\label{ex:sufficient_only}
To Example~\ref{ex:no_ta}, add a coalition $\{d\}$ only to firm $A$
($\mc{C}_A=\{C_1,\{d\}\}$, with $U_A(\{d\})=2>1=U_A(C_1)$ and
$V_d(A,\{d\})=6>V_d(C,C_3)$; everything else as in Example~\ref{ex:no_ta}).
Since $C_1\cap\{d\}=\emptyset$, firm $A$ has no anchor worker, and the preference
cycle of Example~\ref{ex:no_ta} remains, so TA also fails.
Nevertheless, $\mu(A)=\{d\}$ (all other agents unassigned) is a stable matching.
\end{example}

TA holds naturally when workers are shared resources evaluated on a common
scale such as congestion, delay, or resource consumption.
As one such example, consider vehicle-route assignment for traffic safety.
Represent the road network as an undirected graph $\langle \mc{V}, \mc{E} \rangle$,
and define slots $(v,t) \in \mc{V}\times\mc{T}$ representing arrival at
intersection $v$ at time $t$.
The components of the framework then correspond as follows.

\begin{center}
\small
\begin{tabular}{lp{0.52\columnwidth}}
\hline
This framework & Traffic application \\
\hline
Firm $f \in \mc{F}$ & Vehicle \\
Worker $w \in \mc{W}$ & Spatio-temporal slot $(v,t)$ \\
Coalition $C \in \mc{C}_f$ & Feasible route (sequence of slots) \\
Anchor worker $a_f$ & Departure slot $(O_f, T^0_f)$ \\
Firm utility $U_f(C)$ & $-T(C)$ (negative arrival time) \\
Worker utility $V_w(f,C)$ & $-H(C)$ (negative total pedestrian count along route $C$) \\
\hline
\end{tabular}
\end{center}

Two vehicles \emph{collide} if their routes occupy the same 
slot; a feasible route must be collision-free with the routes of the
other matched vehicles.
The vehicle utility $-T(C)$ and the slot utility $-H(C)$ are independent and may
involve a trade-off, in which sense this constitutes an instance of the
firm-worker misalignment.
Since $T(C)$ and $H(C)$ are bounded, a constant shift
$U_f(C):=\text{const}-T(C)$ meets the normalization $U_f,V_w\ge0$ of
Section \ref{sec:notation} without changing preference orders among coalitions or
the truth value of TA, so we retain the notation $-T(C)$ etc. below.
Since the departure slot $(O_f, T^0_f)$ of each vehicle is necessarily contained
in every feasible route, the anchor condition holds automatically by construction.
Here $V_w(f,C)=-H(C)$ is a function common to all slots---the total number of pedestrians encountered over the entire
route $C$---so $-H$ itself is a metaranking and
TA holds automatically by construction (see Section \ref{sec:alignment}).

The same correspondence also holds for service chaining in distributed AI
services, which exhibits the complementarity that the service is realized only
if all processing functions are secured (see Appendix \ref{sec:service-chaining}).
Other settings in which firms' and workers' evaluations of the same allocation
conflict, such as labor-management conflict in wage negotiation, are also widely
covered by this framework.

\section{Algorithm}
\label{sec:algorithm}

In this section, under the anchor condition (Assumption~\ref{asmp:anchor}) and TA (Definition~\ref{def:ta}), we propose Combinatorial Deferred Acceptance with Reproposals (CDAR), an algorithm that terminates and returns a stable matching.

\subsection{Algorithm}

CDAR is based on the firm-proposing DA algorithm; the essential difference from
standard DA is that it permits the reproposals that arise under complementary
preferences.
As shown in Example~\ref{ex:no_anchor}, under complementary preferences a
coalition rejected at one point may later become acceptable.

Let the proposal sequence $\Phi:\mc{F}\to\mc{F}$ be an infinite sequence in which
every firm appears infinitely often.
After initializing all agents as unassigned, while there exists a firm with room
for improvement, the following is repeated: the firm $f$ indicated by $\Phi$
proposes its best not-yet-proposed coalition $C$, and records as the rejection
set $S_f^C$ the workers who do not prefer $(f,C)$ to their current allocation
(line~\ref{alg:reject}).
If no rejection occurs ($S_f^C=\emptyset$), then after releasing the workers of
$f$'s old coalition not contained in $C$, and the allocations of rival firms
that employed workers in $C$ (line~\ref{alg:unmatch}), the allocation $(f,C)$
is newly registered.

The key difference from standard DA lies in the reproposal rule
(line~\ref{alg:repropose-begin}).
A coalition $C'\in R_{f'}$ whose proposal was rejected is removed from the
proposed list $R_{f'}$---thereby becoming eligible for reproposal---at the
moment all workers in the rejection set $S_{f'}^{C'}$ have become free
(unassigned). Reproposal is not permitted while even one rejector remains
assigned: that worker may prefer their current allocation to $(f',C')$, and
permitting a reproposal would allow a blocking pair to persist.

\begin{algorithm}
\caption{Combinatorial Deferred Acceptance with Reproposals (CDAR)}
\label{alg:gcdar}
\begin{algorithmic}[1]

\State $R_f\gets\emptyset$,\ $\mu(f)\gets\emptyset$\ ($\forall f\in\mc{F}$);
\quad $\mu(w)\gets\emptyset$\ ($\forall w\in\mc{W}$)
\Comment{Initialization}

\While{$\exists f\in\mc{F}$ \textbf{s.t.}
$U_f(\mu(f)) < \max_{C\in\mc{C}_f\setminus R_f} U_f(C)$}

    \State $f\gets\Phi.\mathrm{next}$
    \If{$U_f(\mu(f)) \ge \max_{C\in\mc{C}_f\setminus R_f} U_f(C)$}
        \State \textbf{continue}
    \EndIf

    \State $C\gets \arg\max_{C\in\mc{C}_f\setminus R_f} U_f(C)$

    \State $S_f^C \gets \{w\in C \mid V_w(f,C) \le V_w(\mu(w),\mu(\mu(w)))\}$
    \Comment{Record the rejection set}
    \label{alg:reject}

    \State $R_f\gets R_f\cup\{C\}$
    \Comment{Record the proposal}

    \If{$S_f^C=\emptyset$}
    \Comment{Acceptance}
        \State $\mu(w)\gets\emptyset$\ ($\forall w\in\mu(f)\setminus C$)
        \Comment{Release the old allocation}
        \For{$f'\in\{\mu(w)\mid w\in C\}\setminus\{\emptyset\}$}
            \State $\mu(w')\gets\emptyset$\ ($\forall w'\in\mu(f')$);
            \quad $\mu(f')\gets\emptyset$
            \Comment{Unmatch rival firms}
            \label{alg:unmatch}
        \EndFor
        \State $\mu(f)\gets C$;\quad $\mu(w)\gets f$\ ($\forall w\in C$)
        \Comment{Register the new allocation}
    \EndIf

    \For{$f'\in\mc{F}$,\ $C'\in R_{f'}$}
        \label{alg:repropose-begin}
        \State $S_{f'}^{C'}\gets S_{f'}^{C'}\setminus\{w\in S_{f'}^{C'}\mid \mu(w)=\emptyset\}$
        \State \textbf{if} $S_{f'}^{C'}=\emptyset$ \textbf{then}
        $R_{f'}\gets R_{f'}\setminus\{C'\}$
        \Comment{Permit reproposal}
        \label{alg:repropose-end}
    \EndFor
\EndWhile
\State \Return $\mu$

\end{algorithmic}
\end{algorithm}

\subsection{Convergence and Stability}

\begin{theorem}[Convergence]
\label{thm:converge}
CDAR terminates in finite time.
\end{theorem}

The convergence proof uses an $N$-digit potential function based on the ranks
of the metaranking.
Take a rank function $\chi:\mc{A}\to\{1,\dots,m\}$ that numbers the equivalence
classes of $\succeq_\chi$ from the bottom as $1,\dots,m$
(where $m$ is the number of equivalence classes, with
$m\le|\mc{A}|=\sum_{f\in\mc{F}}|\mc{C}_f|$).
For each firm $f\in\mc{F}$, define
\[
P_f(\mu):=
\begin{cases}
\chi(f,\mu(f)) & \mu(f)\neq\emptyset\\[2pt]
0 & \mu(f)=\emptyset
\end{cases}
\]
and let $P_{(1)}\ge P_{(2)}\ge\dots\ge P_{(N)}$ denote $\{P_f(\mu)\}_{f\in\mc{F}}$
arranged in nonincreasing order. Define the potential function as
\[
\phi(\mu)
=
\sum_{k=1}^{N}
P_{(k)}\,(m+1)^{N-k}.
\]
That is, $\phi(\mu)$ interprets $(P_{(1)},\dots,P_{(N)})$ as an $N$-digit
number in base $(m+1)$ (Fig.~\ref{fig:TAnoAnchor}).
Each digit $P_{(k)}$ takes integer values in $\{0,1,\dots,m\}$ and is weighted
in base $(m+1)$, so no carrying occurs and the value of a higher digit always
dominates the total sum of all lower digits.
Moreover, $\phi\ge0$ and $\phi$ is bounded above.

\begin{lemma}
\label{lem:potential}
In each iteration of the algorithm, if $\mu$ is updated, then $\phi(\mu)$ strictly increases.
\end{lemma}

\begin{proof}[Proof sketch]
By the anchor condition, any two feasible coalitions of firm $f$ share $a_f$,
and upon acceptance (line~\ref{alg:reject}) $a_f$ strictly improves; hence by
Definition~\ref{def:ta} no change that decreases $P_f$ is accepted.
Hence a decrease in $P_f$ can occur only when a worker is taken away by another
firm (line~\ref{alg:unmatch}).
In that case the transferred worker belongs to both the old and the new
allocation, so by that worker's strict improvement and Definition~\ref{def:ta},
the digit of the firm that took the worker strictly exceeds the pre-update
digit of the firm that lost the worker.
Therefore, a decrease in the $k$-th digit is necessarily accompanied by a
strict increase in some higher digit $\ell<k$, and $(P_{(1)},\dots,P_{(N)})$
strictly increases in lexicographic order.
The complete proof is given in Appendix \ref{sec:proof-potential}.
\end{proof}

Combined with the boundedness of $\phi$ from above,
this completes the proof of Theorem~\ref{thm:converge}.

\begin{proof}[Proof of Theorem~\ref{thm:converge}]
By Lemma~\ref{lem:potential}, $\phi(\mu)$ strictly increases each time $\mu$ is updated.
Since $\phi$ is bounded above, CDAR terminates after finitely many updates.
\end{proof}

\begin{theorem}[Stability]
\label{thm:stable}
CDAR returns a stable matching $\mu$.
\end{theorem}

\begin{proof}[Proof sketch]
By Theorem~\ref{thm:converge}, CDAR terminates and returns a matching $\mu$; it
remains to show that $\mu$ is stable. Suppose, for contradiction, that a
blocking pair $(f,C)$ exists.
Since $\mu(f)\neq C$, firm $f$ once proposed $C$ and was rejected, so $C\in R_f$.
By the reproposal rule, while $C\in R_f$ the rejection set $S_f^C$ only shrinks
monotonically, so there exists a worker $w\in S_f^C$ who remains assigned until
the algorithm terminates.
Since $w$ never accepts an allocation change that strictly decreases their
utility, $V_w(\mu(w),\mu(\mu(w)))\ge V_w(f,C)$ holds,
contradicting the blocking condition $V_w(f,C)>V_w(\mu(w),\mu(\mu(w)))$.
The complete proof is given in Appendix \ref{sec:proof-stable}.
\end{proof}

Theorem~\ref{thm:stable} implies Theorem~\ref{thm:existence}.

\subsection{Computational Complexity}
\label{sec:complexity}

The complexity of CDAR is evaluated as the product of the number of iterations
and the cost of each iteration.
By Lemma~\ref{lem:potential}, $\phi(\mu)$ strictly increases in each iteration.
Since $\phi(\mu)$ is a \emph{sorted} $N$-digit number, the number of iterations
is bounded by $\binom{N+m}{N}$.
Each iteration consists of three operations: selecting the best not-yet-proposed
coalition, updating the rejection sets, and checking reproposals.
For general problem instances $|\mc{C}_f|$ can be exponential, so selecting the
best not-yet-proposed coalition can be the dominant cost.

On the other hand, in applications where the feasible coalitions of a firm are
implicitly defined as a set of paths on a graph---spatio-temporal routes in
traffic, or node sequences in service chaining---each iteration can be processed
efficiently without explicit enumeration by applying Yen's
algorithm~\cite{yen1971finding}, which sequentially enumerates the best
not-yet-proposed paths.
Furthermore, in the traffic application (see Section \ref{sec:model}), we show that
$|\mc{C}_f|$ itself is of polynomial order under the following structural
assumption.

\begin{assumption}[Departure-only collision resolution]
\label{asmp:departure-only}
When a collision occurs along a route,
it is resolved only by delaying the departure time at the anchor slot
(the departure slot $(O_f,T_f^0)$), and no waiting occurs at intermediate
intersections along the route.
\end{assumption}

Since the elements of $\mc{C}_f$ are slot sequences (pairs of a node sequence and
a departure time), even an identical node sequence (route) can be counted as
distinct elements when the departure delays differ. We therefore let $K$ denote
an upper bound on the number of routes of every vehicle.
By Assumption~\ref{asmp:departure-only}, there is no need to enumerate the
combinations of stopping patterns at intermediate intersections along a route
(in general exponential in the route length). Moreover, since collisions can
always be resolved by departure delays, the candidates to be tried in the
proposal phase are limited to the $K$ routes and their feasibilization.

\begin{theorem}[Complexity in the traffic application]
\label{thm:traffic-complexity}
Under Assumption~\ref{asmp:departure-only}, the total complexity of CDAR is
\[
O\!\left(K(|\mc{V}|^3+N|\mc{V}|)\cdot N^{m}\right),
\]
which is polynomial in $N$ for fixed $m$.
\end{theorem}

The proof is given in Appendix~\ref{sec:proof-traffic-complexity}.
As Theorem~\ref{thm:traffic-complexity} shows, CDAR is slice-wise polynomial
(XP) in $m$.
In particular, since $H(C)$ is an integer count whose range is determined by the
network and independent of the number of vehicles $N$, the value of $m$ does not
grow with $N$, so on any fixed network CDAR runs in time polynomial in $N$.
Furthermore, since a feasible solution exists whenever $K\geq 1$, CDAR can
control its computational cost without sacrificing theoretical correctness and
without exploring all routes.

Also,
the following shows that CDAR in fact functions on, and thus subsumes,
the class of \citet{pycia2012stability}.

\begin{proposition}
\label{prop:cdar-pycia}
Given a many-to-one matching instance that satisfies PA in the sense of
\citet{pycia2012stability} (i.e., for which a metaranking $\chi$ exists),
CDAR converges in finite time and outputs a pairwise stable matching.
\end{proposition}

See Appendix \ref{sec:proof-cdar-pycia} for the proof. The proof proceeds by a
reduction from problem instances of \citet{pycia2012stability} to problem instances of
this paper.
In particular, when the metaranking of \citet{pycia2012stability} can be computed in
polynomial time for each coalition (e.g., when a division rule such as equal
split is used), the reduction is polynomial-time, showing that the problem class
of this paper is at least as hard as that of \citet{pycia2012stability}.

\section{Pareto Optimality}
\label{sec:pareto}

In this section, we show that under strict preferences, any stable
matching is Pareto optimal within the set of stable matchings---and in
particular that this holds for the output of CDAR.
This property holds for stable matchings in general; it extends the known result
for substitutable preferences (\citet{gale1962college,roth1990two}) to the
setting of this paper, which permits complementarity.

Define the payoff of each agent $i\in\mc{I}$ under a matching $\mu$ uniformly as
\[
\pi_i(\mu) =
\begin{cases}
U_f(\mu(f)) & \text{if } i=f\in\mc{F},\\
V_w(\mu(w),\mu(\mu(w))) & \text{if } i=w\in\mc{W}.
\end{cases}
\]

\begin{definition}[Pareto domination]
\label{def:pareto_dom}
A matching $\mu'$ Pareto dominates a matching $\mu$ if
$\pi_i(\mu')\ge\pi_i(\mu)$ for all $i\in\mc{I}$,
and $\pi_j(\mu')>\pi_j(\mu)$ for some $j\in\mc{I}$.
In this case, $j$ is called a \emph{dominating agent} of $\mu$.
\end{definition}

\begin{definition}[Pareto optimality]
\label{def:pareto_opt}
A matching $\mu$ is \emph{Pareto optimal} if there exists no other matching that Pareto dominates $\mu$.
\end{definition}

In the following, we restrict attention to the case in which both $\mu$ and $\mu'$ are stable matchings ($\mu,\mu'\in\mc{M}$).

By Theorem~\ref{thm:stable}, the output of CDAR has no blocking pair.
Therefore, a violation of Pareto optimality---i.e., the existence of another
stable matching that Pareto dominates the output of CDAR---can arise only in the
presence of ties. The following lemma formalizes this.

\begin{lemma}
\label{lem:pareto}
Let $\mu,\mu'\in\mc{M}$ be stable matchings and suppose $\mu'$ Pareto dominates $\mu$.
Let $i$ be a dominating agent of $\mu$.
If $i$ is a worker, then setting $j:=\mu'(i)$, there exists some $k\in\{j\}\cup\mu'(j)$ with $\pi_k(\mu')=\pi_k(\mu)$.
If $i$ is a firm, then there exists some $k\in\mu'(i)$ with $\pi_k(\mu')=\pi_k(\mu)$.
\end{lemma}

The proof is given in Appendix \ref{sec:proof-pareto}.
Consequently, under strict preferences CDAR guarantees Pareto optimality.

\begin{theorem}[Pareto optimality]
\label{thm:pareto}
If there are no ties in the agents' payoffs,
any stable matching $\mu\in\mc{M}$ is Pareto optimal within $\mc{M}$.
In particular, the output of CDAR also satisfies this property.
\end{theorem}

\begin{proof}
When no ties exist, by Lemma~\ref{lem:pareto} no dominating agent exists.
Hence no stable matching Pareto dominates $\mu$, and $\mu$ is Pareto optimal.
\end{proof}

When ties exist, a stable matching may fail to be Pareto optimal (see Appendix \ref{sec:non_pareto}).

One way to mitigate ties is to use the extended payoffs
\begin{align*}
\tilde{\pi}_f(\mu) &= \pi_f(\mu)\cdot(V_{\max}+1) + \pi_{a_f}(\mu),\\
\tilde{\pi}_w(\mu) &= \pi_w(\mu)\cdot(U_{\max}+1) + \pi_{\mu(w)}(\mu)
\end{align*}
(where $U_{\max}=\max_{f\in\mc{F},C\in\mc{C}_f}U_f(C)$).
This is a two-digit value that gives first priority to the original payoff and
breaks ties by a secondary utility value, and it is effective in reducing ties in
practice. However, even this extended payoff does not necessarily remove all ties.

\section{Numerical Verification in the Traffic Application}
\label{sec:experiments}
 
Here we synthesize traffic instances and numerically verify the effectiveness of
our algorithm and its applicability to real problems.
 
\paragraph{Experimental setup.}
On a $5\times 5$ grid road network (25 intersections, 40 links),
we generated synthetic data with numbers of vehicles
$N\in\{5,10,20,30,40,50,100\}$ and 10 random seeds.
The pedestrian count of each link was drawn from the uniform distribution on $\{0,1,\dots,10\}$.
Following the correspondence to the traffic application in Section \ref{sec:model},
the departure slot $(O_f, T^0_f)$ of each vehicle $f$ was routed through a
dedicated dummy node $d_{O_f}$ connected by a dummy edge of length 0 and
pedestrian count 0, thereby guaranteeing the anchor condition by construction.
Each OD pair was chosen uniformly at random from the intersection set $\mc{V}$,
subject to the shortest path between the two having at least 2 links.
Feasible coalitions were defined by shifting departure times whenever
collisions at intersections occurred.
The utilities were $U_f(C)=-T(C)$ and $V_w(f,C)=-H(C)$.
Since $-H$ is a metaranking, TA holds.
No upper limit was imposed on the number of candidate routes $K$.
 
\paragraph{Results.}
For all $N$ and all seeds, CDAR terminated in finite time and output a stable
matching.
Figure~\ref{fig:tcdar} compares the running time of CDAR with the upper bound of
Theorem~\ref{thm:traffic-complexity}.
The theoretical curve is drawn with $K=1$ and $|\mc{V}|=25$, using $m=80$,
an upper bound on the number of distinct cost levels obtained as the maximum
route length ($8$) times the per-link pedestrian cap ($10$). Its proportionality constant was obtained by a least-squares fit as
$c\approx 3.41\times10^{-31}$; since $K\geq 1$ in practice, this curve is
stricter than the theoretical upper bound.
For $N\geq 30$ the running time falls below the bound.
Moreover, fitting the actual number of iterations yields $O(N^{4.8})$,
below the bound $O(N^{m})$.
Furthermore, as an ablation, we compared with a DA in which reproposals are
disabled (removing lines~\ref{alg:repropose-begin}--\ref{alg:repropose-end} of
Algorithm~\ref{alg:gcdar}).
Since this DA proposes each coalition at most once, it always terminates, but
its output retains blocking pairs.
Figure~\ref{fig:nblock} shows their number, which increases monotonically with
$N$.
By contrast, CDAR always has $0$ residual blocking pairs by
Theorem~\ref{thm:stable}, which demonstrates the effectiveness of the reproposal
rule.
 
\begin{figure}[t]
\centering
\begin{subfigure}[t]{.49\columnwidth}
  \centering
  \includegraphics[width=\linewidth]{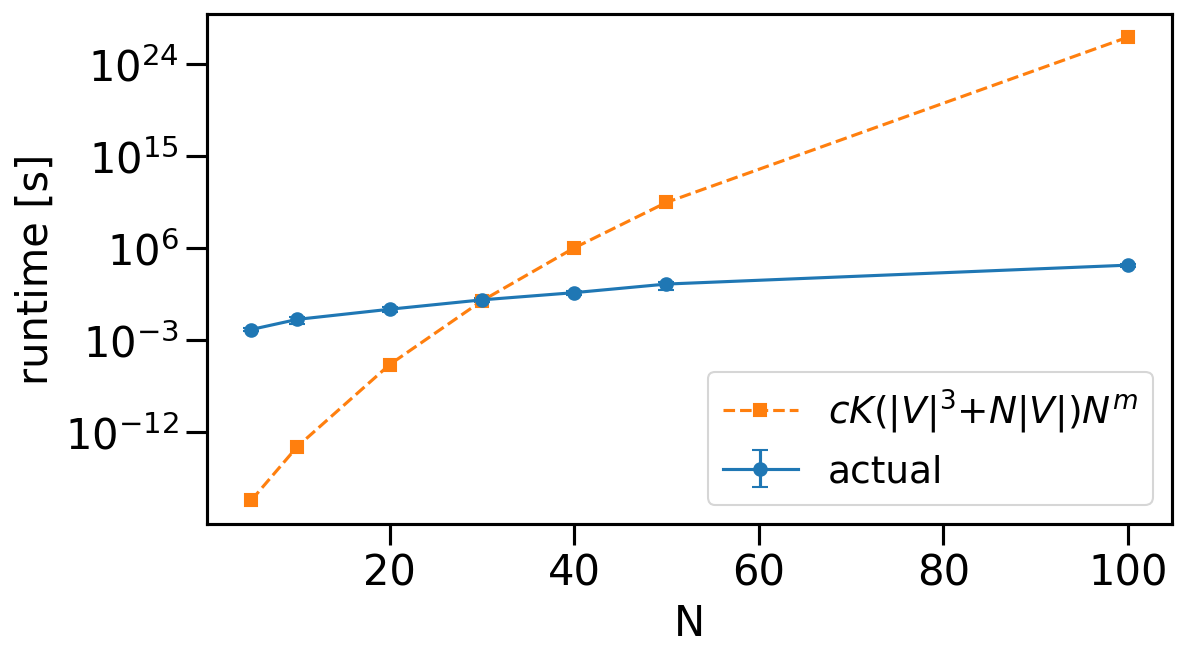}
  \subcaption{Running time of CDAR vs.\ the theoretical upper bound ($K=1$,
  $|\mc{V}|=25$, m=80, constant $c\approx 3.41\times10^{-31}$). Error bars show 95\%
  confidence intervals.}
  \label{fig:tcdar}
\end{subfigure}\hfill
\begin{subfigure}[t]{.49\columnwidth}
  \centering
  \includegraphics[width=\linewidth]{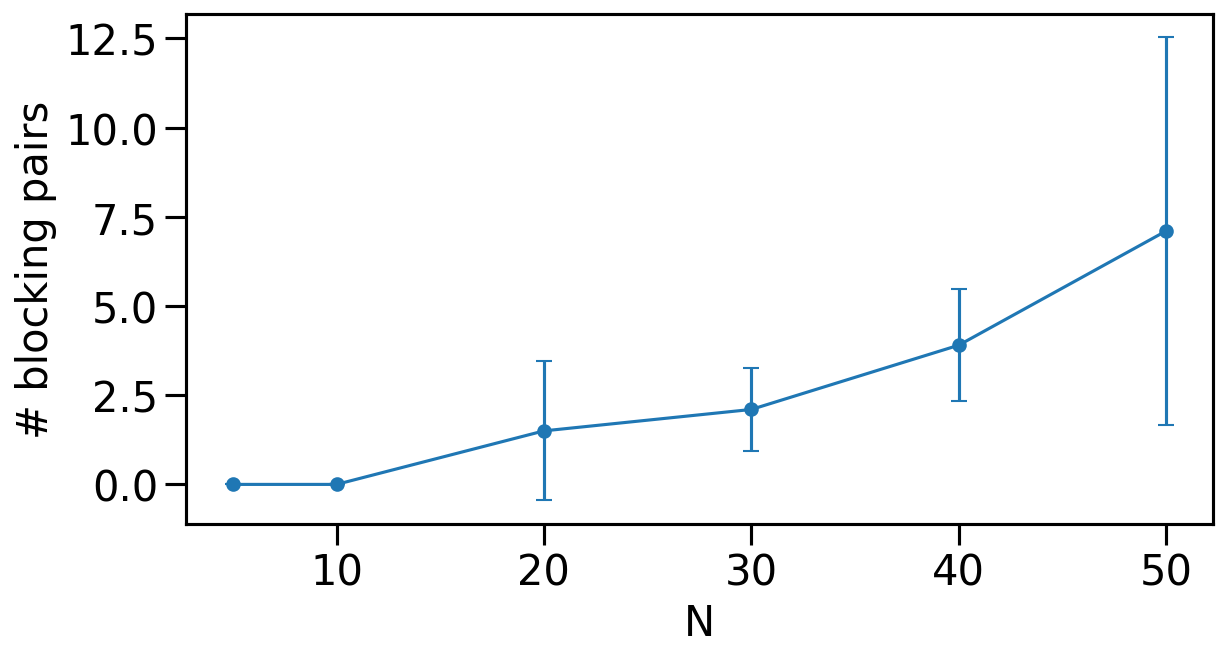}
  \subcaption{Number of blocking pairs remaining in the output of DA with
  reproposals disabled. Error bars show 95\%
  confidence intervals.}
  \label{fig:nblock}
\end{subfigure}
\caption{Experimental results on traffic instances over a $5\times 5$ grid road
network. (a) running time; (b) necessity of reproposals.}
\label{fig:experiments}
\end{figure}

\section*{Conclusion}

In this paper, we studied the many-to-one stable matching problem in the
firm-worker misalignment setting.
Although the existence of stable matchings is not guaranteed in general under
complementary preferences, we showed that imposing the anchor condition
(Assumption~\ref{asmp:anchor}) and Transitive Alignment (Definition~\ref{def:ta})
guarantees the existence of at least one stable matching for every problem
instance.
We also proposed CDAR, a DA-style algorithm permitting reproposals, and proved
finite-time convergence and the output of a stable matching by constructing an
$N$-digit potential function.
Furthermore, we extended the known results on Pareto optimality to the setting
permitting complementarity, showing that under strict preferences the output of
CDAR is Pareto optimal among stable matchings.
In addition, through numerical verification on a $5\times 5$ grid road network,
we confirmed that the anchor condition and TA hold by construction in the traffic
application, and that CDAR outputs stable matchings.

Directions for future work include the following.
First, the analysis of strategic manipulability and fairness.
Second, expansion to broader application domains such as service chaining and
wage negotiation, and validation with real data.

\bibliography{ref}

\clearpage
\begin{appendices}


\color{\newcolor}

\section{Appendix: Proof of Lemma~\ref{lem:potential}}
\label{sec:proof-potential}

\begin{proof}
By the anchor condition, every coalition assigned to firm $f$ contains the anchor worker $a_f$.
By line~\ref{alg:reject} of Algorithm~\ref{alg:gcdar}, $a_f$ strictly improves upon acceptance, and hence by Definition~\ref{def:ta} no change to a coalition that does not increase $P_f$ is accepted.
Therefore, a decrease in $P_f$ can occur only when the firm loses its coalition after being defeated in competition with another firm (line~\ref{alg:unmatch}).

Let $\mu'$ and $\mu$ denote the matchings before and after an iteration, respectively,
and let $f'_{(k)}$ and $f_{(k)}$ denote the firms corresponding to the $k$-th digit of the potential before and after the iteration, respectively.
Suppose $P_{(k)} < P'_{(k)}$; then $P_{f_{(k)}}(\mu') > P_{(k)}$.
This means that firm $f_{(k)}$ was defeated in competition and lost at least one worker $w$
to some firm $f_{(\ell)}$. Since $w$ strictly improves upon acceptance,
$V_w(f_{(\ell)},\mu(f_{(\ell)})) > V_w(f_{(k)},\mu'(f_{(k)}))$,
and since $w$ belongs to both the old and the new allocation,
Definition~\ref{def:ta} yields
\begin{align}
P_{(\ell)}
&=\chi(f_{(\ell)},\mu(f_{(\ell)}))
\notag\\
&>\chi(f_{(k)},\mu'(f_{(k)}))
=P_{f_{(k)}}(\mu')
>P_{(k)}, \label{eq:ta_applied}
\end{align}
and hence $\ell < k$.
Furthermore, since $a_{f_{(\ell)}}$ accepts only updates that strictly increase its utility,
$P_{(\ell)} > P_{f_{(\ell)}}(\mu')$, i.e., $P_{(\ell)} > P'_{(\ell)}$ holds.
Consequently, whenever $P_{(k)}$ decreases, some $P_{(\ell)}$ with $\ell < k$ increases,
and $(P_{(1)},\dots,P_{(N)})$ strictly increases in lexicographic order.
\end{proof}

\section{Appendix: Proof of Theorem~\ref{thm:stable}}
\label{sec:proof-stable}

\begin{proof}
By contradiction. Suppose a blocking pair $(f,C)$ exists. That is,
\begin{align*}
U_f(C) &> U_f(\mu(f)), \\
V_w(f,C) &> V_w(\mu(w),\mu(\mu(w))), \quad\forall w\in C.
\end{align*}
Since $\mu(f)\neq C$, firm $f$ must have once proposed $C$ and been rejected, so $C\in R_f$.
Let $\mu'$ denote the matching at the moment of the last rejection; then the rejection set
\[
S_f^C
=
\{w\in C\mid V_w(\mu'(w),\mu'(\mu'(w)))\ge V_w(f,C)\}
\]
is nonempty.
While $C\in R_f$, firm $f$ does not repropose $C$,
so $S_f^C$ only shrinks monotonically.
Therefore, there exists a worker $w\in S_f^C$ who remains assigned until the algorithm terminates.
Since workers in $S_f^C$ never accept an allocation change that strictly decreases their utility,
\[
V_w(\mu(w),\mu(\mu(w)))
\ge
V_w(\mu'(w),\mu'(\mu'(w)))
\ge
V_w(f,C).
\]
This contradicts $V_w(f,C)>V_w(\mu(w),\mu(\mu(w)))$.
\end{proof}

\section{Appendix: Correspondence to Service Chaining}
\label{sec:service-chaining}

The correspondence of this framework also holds for service chaining in
distributed AI services.
A service chain is an assignment of a sequence of processing functions to a
sequence of nodes on a network graph,
and exhibits the complementarity that the service is realized only if all
functions are secured.

\begin{center}
\small
\begin{tabular}{lp{0.52\columnwidth}}
\hline
This framework & Service chaining application \\
\hline
Firm $f \in \mc{F}$ & Service request \\
Worker $w \in \mc{W}$ & Network node/link \\
Coalition $C \in \mc{C}_f$ & Feasible chain (node sequence) \\
Anchor worker $a_f$ & Ingress node (entry point) \\
Firm utility $U_f(C)$ & $-L(C)$ (negative latency) \\
Worker utility $V_w(f,C)$ & $-\rho(C)$ (negative resource consumption) \\
\hline
\end{tabular}
\end{center}

Since the ingress node of each service is necessarily contained in every
feasible chain, the anchor condition holds automatically by construction.
As for Transitive Alignment, it holds automatically when the resource
consumption of node $w$ is determined solely by the chain $C$, independently of
the type of request $f$---for instance, when the physical capacity of links is
dominant, as with bandwidth consumption, and differences among requests are
negligible---since the common score $-\rho$ itself is then a metaranking
(see Section \ref{sec:alignment}).

\section{Appendix: Proof of Proposition~\ref{prop:cdar-pycia}}
\label{sec:proof-cdar-pycia}

\begin{proof}
In the setting of \citet{pycia2012stability}, a metaranking $\chi$ exists and
each worker's preferences over allocations agree with the order of
$\chi(\{f\}\cup C)$.
Therefore, adopting the order induced by $\chi(\{f\}\cup C)$ directly as the
metaranking of Definition~\ref{def:ta} satisfies TA.
Since $f\in\{f\}\cup C$ holds for every $C\in\mc{C}_f$,
setting $a_f=f$ automatically satisfies the anchor condition
(Assumption~\ref{asmp:anchor}).
Consequently, Lemma~\ref{lem:potential}, Theorem~\ref{thm:converge}, and
Theorem~\ref{thm:stable} apply, and CDAR converges in finite time and outputs a
stable matching.
\end{proof}

\section{Appendix: Proof of Lemma~\ref{lem:pareto}}
\label{sec:proof-pareto}

\begin{proof}
By contradiction. Let $i$ be a dominating agent.
We have $\pi_i(\mu')>\pi_i(\mu)$ and $\pi_k(\mu')\ge\pi_k(\mu)$, $\forall k\in\mc{I}$.

When $i$ is a worker, set $j:=\mu'(i)$.
Suppose $\pi_k(\mu')\neq\pi_k(\mu)$ for all $k$ in $\{j\}\cup\mu'(j)$;
then $\pi_k(\mu')>\pi_k(\mu)$, $\forall k\in\{j\}\cup\mu'(j)$ follows.
Hence
$(j,\mu'(j))$ is a blocking pair for $\mu$ (Definition~\ref{def:blocking}), contradicting the stability of $\mu$.

When $i$ is a firm, suppose $\pi_k(\mu')\neq\pi_k(\mu)$ for all $k$ in $\mu'(i)$;
then $\pi_k(\mu')>\pi_k(\mu)$, $\forall k\in\mu'(i)$ follows.
Hence $(i,\mu'(i))$ is a blocking pair for $\mu$, contradicting the stability of $\mu$.
\end{proof}

\section{Appendix: Violation of Pareto Optimality Caused by Ties}
\label{sec:non_pareto}

\begin{example}[Violation of Pareto optimality caused by ties]
\label{ex:pareto_tie}
Let the set of firms be $\mc{F}=\{A,B,C,D\}$ and the set of workers $\mc{W}=\{a,b,c,d\}$.
The feasible coalitions are
\begin{align*}
&\mc{C}_A=\{\{b\},\{a\}\},\quad \mc{C}_B=\{\{a\},\{b\}\},\\
&\mc{C}_C=\{\{b,c\},\{d\}\},\quad \mc{C}_D=\{\{c\}\}.
\end{align*}
Each coalition contains an anchor worker $a_A,a_B,a_C,a_D$,
but we omit the notation in this example.
The firm utilities satisfy $U_A(\{b\})>U_A(\{a\})$, $U_B(\{a\})>U_B(\{b\})$,
$U_C(\{b,c\})>U_C(\{d\})$, $U_D(\{c\})=1$.
The worker utilities contain ties:
\begin{align*}
V_a(A,\{a\})&=V_a(B,\{a\}),\\
V_b(A,\{b\})&=V_b(B,\{b\})=V_b(C,\{b,c\}),\\
V_c(D,\{c\})&>V_c(C,\{b,c\}).
\end{align*}
Running CDAR under the proposal sequence $\Phi=(C,A,D,B,\ldots)$ outputs
\[
\mu(A)=\{a\},\ \mu(B)=\{b\},\ \mu(C)=\{d\},\ \mu(D)=\{c\}.
\]
Although $\mu$ has no blocking pair, setting
\[
\mu'(A)=\{b\},\ \mu'(B)=\{a\},\ \mu'(C)=\{d\},\ \mu'(D)=\{c\},
\]
$\mu'$ has no blocking pair either. Moreover, $\pi_A(\mu')>\pi_A(\mu)$, $\pi_B(\mu')>\pi_B(\mu)$,
and the payoffs of all other agents are unchanged,
so $\mu'$ Pareto dominates $\mu$.
\end{example}

\section{Appendix: Proof of the Complexity in the Traffic Application}
\label{sec:proof-traffic-complexity}

The following lemma shows that any collision arising along a route can be
resolved within a finite number of trials by departure delays alone,
without stopping at intermediate points along the route.
This justifies that the candidates to be tried in the proposal phase are limited
to the $K$ node sequences and their feasibilization,
without having to maintain the stopping patterns at intermediate points along a
route combinatorially as separate candidates.

\begin{lemma}[Feasibility via departure delays]
\label{lem:feasibility}
Under Assumption~\ref{asmp:departure-only},
for any node sequence $\text{path}(C)$ (of length $r\le|\mc{V}|$),
increasing the departure delay at most $O(rN)$ times yields a slot sequence
$C\in\mc{C}_f$ that resolves all collisions with the other $N-1$ vehicles.
\end{lemma}

\begin{proof}
A collision with another vehicle at a slot along the route occurs only when
that vehicle attempts to occupy the same slot at the same time.
Delaying the departure by one unit shifts the arrival times of all slots along
the vehicle's route uniformly backward.
Since one pair of vehicles can collide separately at each of the multiple
intersections shared by their routes, collisions per vehicle pair can occur at
at most $r$ slots.
Therefore, the total number of collision points with the other $N-1$ vehicles is
at most $O(rN)$, and each collision point can be resolved by increasing the
departure delay at most once (it suffices to delay until the overlap of
occupation times is avoided).
Hence, the number of departure-delay increases required to resolve all
collisions is bounded by $O(rN)\subseteq O(N|\mc{V}|)$.
\end{proof}

\begin{proof}[Proof of Theorem~\ref{thm:traffic-complexity}]
We decompose the proof into three components:
(i) the number of node sequences and the cost of enumerating them,
(ii) the number of delay increases required to turn one node sequence into a
feasible slot sequence, and
(iii) the number of iterations until convergence and the total complexity.

\textbf{(i) Number of node sequences and enumeration cost.}
As defined in the main text, the number of distinct node sequences for firm $f$
is at most $K$.
Yen's algorithm~\cite{yen1971finding} enumerates the node sequences up to the
$k$-th shortest travel time in $O(k|\mc{V}|^3)$,
so the cost of enumerating all $K$ candidate node sequences for firm $f$ is
$O(K|\mc{V}|^3)$.

\textbf{(ii) Number of delay increases to make a node sequence feasible.}
By Lemma~\ref{lem:feasibility}, for a given node sequence
(of length $r\le|\mc{V}|$), the number of departure-delay values tried in order
to resolve collisions is at most $O(rN)\subseteq O(N|\mc{V}|)$.

Since the delay determination of (ii) is performed for each of the $K$ node
sequences enumerated in (i), for firm $f$, combining the enumeration cost
$O(K|\mc{V}|^3)$ of the candidate node sequences with the total delay
determination cost $O(KN|\mc{V}|)$ (i.e., $O(N|\mc{V}|)$ per sequence for the
$K$ sequences), we obtain
\[
O(K|\mc{V}|^3) + O(KN|\mc{V}|) = O(K(|\mc{V}|^3+N|\mc{V}|)).
\]

\textbf{(iii) Number of iterations until convergence and total complexity.}
By Lemma~\ref{lem:potential}, the potential function $\phi(\mu)$ strictly
increases each time $\mu$ is updated.
$\phi(\mu)$ is a sorted $N$-digit number formed by arranging
$\{P_f(\mu)\}_{f\in\mc{F}}$ in descending order, and each digit takes values in
$\{0,1,\dots,m\}$.
Therefore, treating $m$ as a constant, the total number of values $\phi(\mu)$
can take is bounded by
\[
\binom{N+m}{N}
=O\!\left(\frac{(N+m)^{m}}{m!}\right)
=O\!\left(N^{m}\right),
\]
and the number of iterations is bounded by this quantity.
The dominant cost of each iteration is the per-firm candidate enumeration and
feasibilization cost $O(K(|\mc{V}|^3+N|\mc{V}|))$ evaluated in (i) and (ii),
so the total complexity of CDAR is
\begin{align*}
&O(K(|\mc{V}|^3+N|\mc{V}|))\times O\!\left(N^{m}\right)\\
&\quad=O\!\left(K(|\mc{V}|^3+N|\mc{V}|)\,N^{m}\right).
\end{align*}

\end{proof}

\section{Appendix: Details of reproducibility}
 
We call numpy.random.seed at the beginning of the main code with a different seed for each run of the simulation.
The simulation is run on an AWS t2.xlarge instance (4 vCPU, 16GB RAM, Ubuntu Server 24.04 LTS) with Python 3.12.3.
All Python packages and versions are listed in requirements.txt in the code.
Table \ref{tbl:parameter} shows the list of all parameters.
 
\begin{table}[h]
\centering
\begin{tabular}{|l|l|}
\hline
$|\mc{V}|$ & As described in the paper \\
$L_e$ & Length of edge $e\in\mc{E}$: Random number greater than 5 \\
$dT$ & Time width of a slot: 4 \\
$N$ & As described in the paper \\
$V_i^{\max}$ & Vehicle speed: 1 \\
$\Phi$ & Ascending order of $\mc{N}$\\
\hline
\end{tabular}
 
\caption{List of parameters}
\label{tbl:parameter}
\end{table}

\end{appendices}


\end{document}